\documentclass[12pt]{article}
\usepackage{graphicx,amssymb,amsmath,amsthm,url,bm,verbatim}

\theoremstyle{plain}
\newtheorem{theorem}{Theorem}

\theoremstyle{definition}
\newtheorem{definition}{Definition}

\newcommand{\ba}{\mathbf{a}}
\newcommand{\bb}{\mathbf{b}}
\newcommand{\bc}{\mathbf{c}}

\newcommand{\bx}{\mathbf{x}}

\newcommand{\cC}{\mathcal{C}}

\newcommand{\desc}{\mathrm{desc}}

\begin{document}
\title{Probabilistic Existence Results for Separable Codes}
\author{Simon R.\ Blackburn\\
Department of Mathematics\\
Royal Holloway University of London\\
Egham, Surrey TW20 0EX\\
United Kingdom}
\maketitle

\begin{abstract}
Separable codes were defined by Cheng and Miao in 2011, motivated by applications to the identification of pirates in a multimedia setting. Combinatorially, $\overline{t}$-separable codes lie somewhere between $t$-frameproof and $(t-1)$-frameproof codes: all $t$-frameproof codes are $\overline{t}$-separable, and all $\overline{t}$-separable codes are $(t-1)$-frameproof. Results for frameproof codes show that (when $q$ is large) there are $q$-ary $\overline{t}$-separable codes of length $n$ with approximately $q^{\lceil n/t\rceil}$ codewords, and that no $q$-ary $\overline{t}$-separable codes of length $n$ can have more than approximately $q^{\lceil n/(t-1)\rceil}$ codewords.

The paper provides improved probabilistic existence results for $\overline{t}$-separable codes when $t\geq 3$. More precisely, for all $t\geq 3$ and all $n\geq 3$, there exists a constant $\kappa$ (depending only on $t$ and $n$) such that there exists a $q$-ary $\overline{t}$-separable code of length $n$ with at least $\kappa q^{n/(t-1)}$ codewords for all sufficiently large integers $q$. This shows, in particular, that the upper bound (derived from the bound on $(t-1)$-frameproof codes) on the number of codewords in a $\overline{t}$-separable code is realistic.

The results above are more surprising after examining the situation when $t=2$. Results due to Gao and Ge show that a $q$-ary $\overline{2}$-separable code of length $n$ can contain at most $\frac{3}{2}q^{2\lceil n/3\rceil}-\frac{1}{2}q^{\lceil n/3\rceil}$ codewords, and that codes with at least $\kappa q^{2n/3}$ codewords exist. Thus optimal $\overline{2}$-separable codes behave neither like $2$-frameproof nor $1$-frameproof codes.

The paper also observes that the bound of Gao and Ge can be strengthened to show that the number of codewords of a $q$-ary $\overline{2}$-separable code of length $n$ is at most
\[
q^{\lceil 2n/3\rceil}+\tfrac{1}{2}q^{\lfloor n/3\rfloor}(q^{\lfloor n/3\rfloor}-1).
\]

\end{abstract}

\section{Introduction}
\label{sec:introduction}

Let $F$ be a finite set of size $q$, where $q\geq 2$. Let $n$ be an integer, where $n\geq 2$. For a subset $X\subseteq F^n$ of words over $F$ of length $n$, and for $k\in \{1,2,\ldots,n\}$, we define
$X(k)$ to be the set of $k$th components of words in $X$, and we define the set $\desc(X)$ of \emph{descendants} of $X$ by
\[
\desc(X)=X(1)\times X(2)\times\cdots \times X(n).
\]
For example, if $X=\{0000,0111,0012\}$ then $X(1)=\{0\}$, $X(2)=X(3)=\{0,1\}$, $X(4)=\{0,1,2\}$ and $\desc(X)$ is the following set of words:
\[
\{0000,0100,0010,0110,0001,0101,0011,0111,0002,0102,0012,0112\}.
\]
Note that $X\subseteq \desc(X)$.

The notion of a descendant set arose in connection with watermarked multimedia content. In these applications, each user is given a copy of the data watermarked in a different way: these watermarks correspond to a $q$-ary word of length $n$. The watermarks corresponding to users are chosen to lie in some carefully constructed code $\cC$. An adversary receives $t$ copies with different watermarkings (from $t$ compromised users), the watermarks corresponding to a set $X$ of $t$ codewords. The adversary then attempts to construct a new copy of the data. In certain well-motivated models, an adversary can create data with any watermark taken from $\desc(X)$. Various security models and applications have been considered; different classes of codes are defined, depending on the goals of the adversary. See Blackburn~\cite{Blackburn_survey} for a survey of early results in the area; see, for example, Blackburn, Etzion, Stinson and Zaverucha~\cite{BlackburnEtzionStinsonZaverucha08} for more recent work on generalisations. As an example of the combinatorial objects that have been considered in this context, $t$-frameproof codes arose from a paper of Boneh and Shaw~\cite{BonehShaw98} that addresses the problem of preventing an adversary from framing an innocent user:

\begin{definition}
Let $t$ be a positive integer. A code $\cC\subseteq F^n$ is \emph{$t$-frameproof} if for all $X\subseteq\cC$ with $|X|\leq t$, we have $\desc(X)\cap\cC=X$.
\end{definition}

If users' watermarks are taken from a $t$-frameproof code $\cC$, an adversary who obtains copies of the data from $t$ users is unable to construct a watermarked copy of that data that would have been given to another user.

More recently, models have been considered for an `averaging attack' adversary; see Trappe, Wu, Wang and Liu~\cite{TrappeWuWangLiu03}. Here $\desc(X)$, rather than an element of $\desc(X)$, may be extracted from the data produced by the adversary. Cheng and Miao~\cite{ChengMiao11} introduced separable codes in the context of this model:

\begin{definition}
A code $\cC\subseteq F^n$ is a \emph{$\overline{t}$-separable code} if for all distinct $X,X'\subseteq\cC$ with $|X|\leq t$ and $|X'|\leq t$, we have $\desc(X)\not=\desc(X')$.
\end{definition}

Cheng and Miao~\cite{ChengMiao11} pointed out that there is a close relationship between $\overline{t}$-separable codes and frameproof codes:
\begin{theorem}\emph{\cite[Lemmas~4.5 and~4.6]{ChengMiao11}}
\label{thm:frameproof_separating}
Let $t$ be an integer, $t\geq 2$. Any $t$-frameproof code is a $\overline{t}$-separable code. Any $\overline{t}$-separable code is a $(t-1)$-frameproof code.
\end{theorem}
\begin{proof}
Suppose $\cC$ is not a $\overline{t}$-separable code. Let $X$ and $X'$ be distinct subsets of $\cC$ such that $|X|\leq t$, $|X'|\leq t$ and $\desc(X)=\desc(X')$. By interchanging $X$ and $X'$ if necessary, we may assume that $X'\not\subseteq X$. Clearly $X\subseteq\desc(X)$ and $X'\subseteq \desc(X')=\desc(X)$, and so $X\cup X'\subseteq \desc(X)\cap \cC$. The set $X$ shows that $\cC$ is not $t$-frameproof: $\desc(X)$ contains a non-empty set $X'\setminus X$ of codewords not in X. So the first statement of the theorem follows.

Suppose $\cC$ is a code that is not $(t-1)$-frameproof. So there exists $X\subseteq \cC$ and $\bc\in \cC\setminus X$ with $|X|\leq t-1$ and $\bc\in\desc(X)$. Define $X'=X\cup\{\bc\}$. Then the sets $X$ and $X'$ show that $\cC$ is not $\overline{t}$-separable, since $\desc(X)=\desc(X')$. So the theorem follows.
\end{proof}
Cheng and Miao also relate $\overline{t}$-separable codes to matrices arising in combinatorial group testing, and apply these codes to generalisations of the structures used by Trappe \emph{et al.} in their multimedia application.

Several other papers have considered $\overline{t}$-separable codes. Cheng, Ji and Miao~\cite{ChengJiMiao12} provide a detailed analysis of $\overline{2}$-separable codes of lengths~$2$  and~$3$, finding optimal codes in many cases. Cheng, Fu, Jiang, Lo and Miao~\cite{ChengFuJiangLoMiao15} improve the length $2$ results using graph theory and packings derived from projective planes.
Gao and Ge~\cite{GaoGe14} provide probabilistic constructions of, and upper bounds for the number of codewords in, $\overline{2}$-separable codes (which we discuss further below). Finally, various related codes have been studied: secure separable codes~\cite{JiangChengMiao15} and multimedia IPP codes~\cite{ChengFuJiangLoMiaoIPP14,JiangChengMiaoWuIPP14}.

Consider a $q$-ary $\overline{t}$-separable code $\cC$ of length $n$ containing as large a number of codewords as possible. Bounds~\cite{Blackburn03} on frameproof codes, and Theorem~\ref{thm:frameproof_separating} above, imply that there exists a constant $\kappa$, depending only on $n$ and $t$, such that 
\begin{equation}
\label{eqn:frameproof_bounds}
\kappa q^{\lceil n/t\rceil}\leq |\cC|\leq (t-1)q^{\lceil n/(t-1)\rceil}
\end{equation}
for all sufficiently large $q$. If the true optimal value of $|\cC|$ is close to the lower bound, we can think of $\overline{t}$-separable codes as being close to $t$-frameproof codes; if the optimal value of $|\cC|$ is close to the upper bound, we can think of $\overline{t}$-separable codes as being close to $(t-1)$-frameproof codes. This paper asks: Which (if any) of these two situations occur?

The answer to this question in the case $t=2$ follows from the results of Gao and Ge~\cite{GaoGe14}. Firstly they give a probabilistic argument (reproduced in Section~\ref{sec:2_lower_bound} below) which shows that there are at least $\kappa q^{2n/3}$ codewords in a $\overline{2}$-separable code for all sufficiently large $q$, for some constant $\kappa$.  Secondly, they combine a bound on $\overline{2}$-separable codes of length $2$ due to Cheng, Ji and Miao~\cite{ChengJiMiao12} with an observation that $\overline{2}$-separable codes can naturally be used to construct shorter $\overline{2}$-separable codes over a larger alphabet, to prove the following upper bound. Let $\cC$ be a $q$-ary $\overline{2}$-separable code of length $n$. Then
\[
|\cC|\leq \tfrac{3}{2}q^{2\lceil n/3\rceil}-\tfrac{1}{2}q^{\lceil n/3\rceil}.
\]
So the answer to our question when $t=2$ is `Neither situation occurs', as an optimal $\overline{2}$-separable code has about $q^{2n/3}$ codewords for large $q$.

In fact, the upper bound of Gao and Ge can be tightened to improve the exponent in the leading term (when $3$ does not divide $n$). We show (Theorem~\ref{thm:2_upper_bound} below) that a $\overline{2}$-separable code $\cC$ satisfies 
\[
|\cC|\leq q^{\lceil (2/3)n\rceil}+\frac{1}{2}q^{\lfloor (1/3)n\rfloor}(q^{\lfloor (1/3)n\rfloor}-1)\leq (3/2)q^{\lceil (2/3)n\rceil}.
\]

When $t\geq 3$, we will show that the answer to our question is `$\overline{t}$-separable codes are close to being $(t-1)$-frameproof codes'. This follows from the upper bound~\eqref{eqn:frameproof_bounds} from $(t-1)$-frameproof codes, combined with Theorem~\ref{thm:general_lower_bound} below, which says that there exist $q$-ary $\overline{t}$-separable codes of length $n$ with at least $\kappa q^{n/(t-1)}$ codewords for all sufficiently large $q$, where $\kappa$ is a constant depending only on $n$ and $t$. Just like Gao and Ge's argument when $t=2$, Theorem~\ref{thm:general_lower_bound} is probabilistic, using random choice with expurgation. However, the naive generalisation of Gao and Ge's argument will not give a good enough bound: some care is needed when defining the set of `bad' events that should be excluded.

We comment that Gao and Ge~\cite[Corollary~4.8]{GaoGe14} also provide a probabilistic existence result for $\overline{t}$-separable codes for $t\geq 3$. Indeed, they show that a $q$-ary $\overline{t}$-separable code of length $n$ with $M$ codewords exists for some value of $n$ satisfying
\[
n\leq \frac{q^t}{(q-1)^t}\left(1+\ln \binom{M}{1,t}\right).
\]
However, this result is much weaker than Theorem~\ref{thm:general_lower_bound} below for many parameter sets. For example, when $n$ and $t$ are fixed and $q\rightarrow\infty$, the result of Gao and Ge only guarantees the existence of a code with approximately $e^{(n-\kappa)/(t+1)}$ codewords for some constant $\kappa$. 

The structure of the remainder of the paper is as follows. In Section~\ref{sec:upper_bounds}, after briefly providing (Theorem~\ref{thm:general_upper_bound}) a proof of the upper bound~\eqref{eqn:frameproof_bounds} on the number of codewords in a $\overline{t}$-separable code, we state and prove our improvement (Theorem~\ref{thm:2_upper_bound}) of Gao and Ge's upper bound for $t=2$. In Section~\ref{sec:2_lower_bound} we provide a recap of  Gao and Ge's probabilistic existence result (Theorem~\ref{thm:2_lower_bound}) for $\overline{2}$-separable codes. We include the full proof of this result for completeness, to establish notation, and to provide an introduction to the more complicated proof in the general case. This general case (Theorem~\ref{thm:general_lower_bound}) is contained in Section~\ref{sec:general_lower_bound}.

\section{Upper bounds}
\label{sec:upper_bounds}

Theorem~\ref{thm:general_upper_bound} below is due to Cheng, Ji and Miao~\cite[Lemma~3.7 and Theorem~3.8]{ChengJiMiao12}; their proof uses a bound of Blackburn~\cite[Theorem~1]{Blackburn03} on the maximum cardinality of a $(t-1)$-frameproof code. We include a proof here for completeness. For simplicity, we use a slightly weaker bound on the maximum cardinality of frameproof codes due to Staddon, Stinson and Wei~\cite{StaddonStinsonWei01}.

\begin{theorem} 
\label{thm:general_upper_bound}
Let $q$, $n$ and $t$ be positive integers such that $q\geq 2$, $n\geq 2$ and $t\geq 2$.  Let $\cC$ be a $q$-ary $\overline{t}$-separable code of length $n$. 
Then $|\cC|\leq (t-1)q^{\lceil n/(t-1)\rceil}$.
\end{theorem}
\begin{proof}
By Theorem~\ref{thm:frameproof_separating}, $\cC$ is a $(t-1)$-frameproof code.

For a subset $L\subseteq\{1,2,\ldots,n\}$ of components, let $U_L$ be the set of all codewords that are uniquely determined by their components in $L$. (So if $\bc\in U_L$ no other codeword agrees with $\bc$ on all positions in $L$.) Clearly, $|U_L|\leq q^{|L|}$.

Let $L_1,L_2,\ldots,L_{t-1}$ be a partition of $\{1,2,\ldots,n\}$ into $t-1$ subsets of size at most $\lceil n/(t-1)\rceil$. We claim that $\cC=\bigcup_{i=1}^{t-1}U_{L_i}$. Clearly this claim is enough to prove the theorem, for then
\[
|\cC|\leq \sum_{i=1}^{t-1}|U_{L_i}|\leq\sum_{i=1}^{t-1}q^{\lceil n/(t-1)\rceil}=(t-1)q^{\lceil n/(t-1)\rceil}.
\]

Suppose for a contradiction that there exists $\bc\in\cC\setminus \bigcup_{i=1}^{t-1}U_{L_i}$. Then for each $i\in\{1,2,\ldots ,t-1\}$ there exists $\bx_i\in \cC\setminus\{\bc\}$ that agrees with $\bc$ on all components in $L_i$. So, setting $X=\{\bx_1,\bx_2,\ldots,\bx_{t-1}\}$ we see that $\bc\in \desc(X)$. Since $\bc\in \cC\setminus X$, this contradicts the fact that $\cC$ is $(t-1)$-frameproof. So the claim, and therefore the theorem, follows.
\end{proof}

When $t=2$ the bound of Theorem~\ref{thm:general_upper_bound} is vacuous. However, Gao and Ge~\cite{GaoGe14} obtained a non-trivial bound in this case. We tighten their bound as follows.

\begin{theorem}
\label{thm:2_upper_bound}
Let $q$ and $n$ be integers such that $q\geq 2$ and $n\geq 2$. If $\cC$ is a $q$-ary $\overline{2}$-separable code of length $n$, then
\[
|\cC|\leq q^{\lceil (2/3)n\rceil}+\frac{1}{2}q^{\lfloor (1/3)n\rfloor}(q^{\lfloor (1/3)n\rfloor}-1)\leq (3/2)q^{\lceil (2/3)n\rceil}.
\]
\end{theorem}
\begin{proof}
Suppose $\cC$ is a code over the alphabet $F$, where $|F|=q$.

Define $r=\lceil (2/3)n\rceil$ and define $s=\lfloor (1/3)n\rfloor$. Note that $r+s=n$.

For a word $\bx\in F^n$, define $\pi(\bx)\in F^r$ to be the prefix of $\bx$ of length $r$, and define $\sigma(\bx)\in F^s$ to be the suffix of $\bx$ of length $s$.

For $\ba\in F^r$, define $T_\ba\subseteq F^s$ to be the set of suffices of words in $\cC$ beginning with $\ba$. So
\[
T_\ba=\{\sigma(\bx)\mid\bx\in C\cap\pi^{-1}(\ba)\}.
\]
Write $t_\ba=|T_\ba|$. Clearly
\[
|\cC|=\sum_{\ba\in F^r}t_\ba
\]
since $\cC$ is a disjoint union of the codewords contributing to each of the sets~$T_\ba$.

We claim that for any distinct $\ba,\ba'\in F^r$ we have $|T_\ba\cap T_{\ba'}|< 2$. To show this, we assume for a contradiction that there exist distinct $\bb,\bb'\in T_\ba\cap T_{\ba'}$. Then the words $\ba\bb, \ba\bb',\ba'\bb$ and $\ba'\bb'$ all lie in $\cC$. But then the sets $X=\{\ba\bb,\ba'\bb'\}$ and $X'=\{\ba\bb',\ba'\bb\}$ have the same set of descendants, contradicting the fact that $\cC$ is $\overline{2}$-separable. So our claim follows.

Each set $T_\ba$ contains $\binom{t_\ba}{2}$ subsets of $F^s$ of size $2$. The claim above shows that no two subsets  $T_\ba$ and $T_{\ba'}$ can contain the same subset of size $2$. Since $F^s$ contains $\binom{q^s}{2}$ subsets of size $2$,
\[
\sum_{\ba\in F^r}\tfrac{1}{2}t_\ba(t_\ba-1)\leq \tfrac{1}{2}q^s(q^s-1).
\]
The maximum sum of non-negative integer variables $u_\ba$ such that
\[
\sum_{\ba\in F^r}\tfrac{1}{2}u_{\ba}(u_\ba-1)\leq \tfrac{1}{2}q^{s}(q^s-1)
\]
is achieved when $u_{\ba}=2$ for exactly $\frac{1}{2}q^s(q^s-1)$ values of $\ba$, and $u_{\ba}=1$ for the remaining values. Hence
\begin{align*}
|\cC|&=\sum_{\ba\in F^r}t_\ba\leq\sum_{\ba\in F^r} u_\ba\\
&\leq 2\left(\tfrac{1}{2}q^s(q^s-1)\right)+\left(q^r-\tfrac{1}{2}q^s(q^s-1)\right)\\
&=q^r+\tfrac{1}{2}q^s(q^s-1)
\end{align*}
and so the theorem follows.
\end{proof}

\section{A lower bound for $\overline{2}$-separable codes}
\label{sec:2_lower_bound}

This section is an exposition of a lower bound for the cardinality of a $\overline{2}$-separable code due to Gao and Ge~\cite{GaoGe14}. 

\begin{theorem}
\label{thm:2_lower_bound}
Let $q$ and $n$ be integers such that $q\geq 2$ and $n\geq 2$. Then there exists a $q$-ary $\overline{2}$-separable code of length $n$ with $M$ codewords, where
\[
M=\max_{N\in\mathbb{N}}\left\{N-\left\lfloor \binom{N}{2}q^{-n}+3\binom{N}{4}2^nq^{-2n}\right\rfloor\right\}.
\]
In particular, for each value of $n$ there exists a constant $\kappa$ (depending only on $n$) with the following property: For all sufficiently large integers $q$, there exists $q$-ary $\overline{2}$-separable code of length $n$ with at least $\kappa q^{(2/3)n}$ codewords. 
\end{theorem}
\begin{proof}
The theorem follows by a straightforward application of the technique of random choice with expurgation. The details are as follows. Let $N$ be an integer to be chosen later. Choose words $\bc_1,\bc_2,\ldots,\bc_N\in F^n$ uniformly and independently at random.

For $i,j\in\{1,2,\ldots,N\}$ with $i<j$, let $E_{i,j}$ be the event that $\bc_i=\bc_j$. So each event $E_{i,j}$ occurs with probability $q^{-n}$. Let $R_{i,j}$ be the indicator random variable taking the value $1$ when $E_{i,j}$ occurs, otherwise taking the value $0$. The expected value of $R_{i,j}$ is $q^{-n}$.

For distinct $i,j,i',j'\in\{1,2,\ldots,n\}$ with $i<j$, with $i<i'$ and with $i'<j'$, set $X=\{\bc_i,\bc_j\}$ and $X'=\{\bc_{i'},\bc_{j'}\}$. Let $F_{i,j,i',j'}$ be the event that $\desc(X)=\desc(X')$. This happens only if $X(k)=X'(k)$ for all $k\in\{1,2,\ldots ,n\}$. Now, $X(k)=X'(k)$ with probability at most $2/q^2$: once the values of the $k$th components of $\bc_{i'}$ and $\bc_{j'}$ are chosen, the $k$th component of $\bc_{i}$ must agree with one of the $k$th components of $\bc_{i'}$ and $\bc_{j'}$ (this event occurs with probability at most $2/q$), and then there is a unique value that $k$th component of $\bc_j$ must take so that $X(k)=X'(k)$ (and the $k$th component of $\bc_j$ takes this value with probability $1/q$). So $F_{i,j,i',j'}$ occurs with probability at most $(2/q^2)^n$. Let $S_{i,j,i',j'}$ be the indicator random variable taking the value $1$ when $F_{i,j,i',j'}$ occurs, otherwise taking the value $0$. The expected value of $S_{i,j,i',j'}$ is at most $2^nq^{-2n}$. Note that there are $3\binom{N}{4}$ choices for the quadruple $(i,j,i',j')$. To see this, first note that there are $\binom{N}{4}$ choices for the subset $L=\{i,j,i',j'\}$. Once $L$ is fixed, $i$ is determined as the smallest element of~$L$. There are then $3$ choices for $j\in L\setminus\{i\}$. Finally, $i'$ and $j'$ are determined as (respectively) the smallest and largest elements from $L\setminus\{i,j\}$.

Let $Z=\sum_{i,j}R_{i,j}+\sum_{i,j,i',j'}S_{i,j,i',j'}$. So $Z$ is a random variable which counts the number of the events $E_{i,j}$ and $F_{i,j,i',j'}$ that occur. By linearity of expectation, the expected value of $Z$ is at most
\[
\sum_{i,j}q^{-n}+\sum_{i,j,i',j'}2^nq^{-2n}=\binom{N}{2}q^{-n}+3\binom{N}{4}2^nq^{-2n}.
\]
In particular, there is a choice of words $\bc_i$ such that at most $b$ of the events $E_{i,j}$ and $F_{i,j,i',j'}$ occur where
\[
b=\left\lfloor \binom{N}{2}q^{-n}+3\binom{N}{4}2^nq^{-2n}\right\rfloor.
\]
We fix such a choice of words.

We form a set $I\subseteq \{1,2,\ldots,N\}$ by adding $i$ to $I$ whenever $E_{i,j}$ or $F_{i,j,i',j'}$ occurs. So $|I|\leq b$. We define $\cC=\{\bc_i\mid i\in \{1,2,\ldots,n\}\setminus I\}$. By our construction of $I$ and by definition of the event $E_{i,j}$, all the codewords in $\cC$ are distinct (we have removed one of each pair $\{\bc_i,\bc_j\}$ with $\bc_i=\bc_j$). Thus $|\cC|\geq N-b$.

We claim that $\cC$ is a $\overline{2}$-separable code. Let $X,X'\subseteq \cC$ be distinct non-empty subsets of codewords with $|X|\leq 2$ and $|X'|\leq 2$. If $|X|=|X'|=2$, then $\desc(X)\not=\desc(X')$, since by our construction of $I$ and the definition of the events $F_{i,j,i',j'}$ we have removed one of the four codewords involved when we constructed $\cC$ from $\{\bc_1,\bc_2,\ldots,\bc_N\}$. So, without loss of generality, we may assume $|X'|=1$, which implies that $|\desc(X')=1$. But then $\desc(X)=\desc(X')$ implies that $X=X'$ in this situation, contradicting the fact that $X$ and $X'$ are distinct. So $\cC$ is $\overline{2}$-separable. Thus the first statement of the theorem follows.

To prove the last statement of the theorem, we fix $n$ and choose $N=\lfloor 2^{-(1/3)n}q^{(2/3)n}\rfloor$. We observe that
\begin{align*}
|\cC|&=N-\left\lfloor \binom{N}{2}q^{-n}+3\binom{N}{4}2^nq^{-2n}\right\rfloor\\
&\geq N-N^2q^{-n}-\tfrac{3}{24}N^42^nq^{-2n}\\
&\geq (2^{-(1/3)n}-2^{-(2/3)n}q^{-(1/3)n}-\tfrac{3}{24}2^{-(1/3)n}-o(1))q^{(2/3)n}\\
&>\tfrac{1}{2}2^{-(1/3)n}q^{(2/3)n}
\end{align*}
when $q$ is sufficiently large. So the final statement of the theorem follows with $\kappa=\frac{1}{2}2^{-(1/3)n}$. We remark that we have not tried to optimise the constant $\kappa$ here: see Gao and Ge~\cite{GaoGe14} for a tight value. 
\end{proof}

\section{A lower bound for $\overline{t}$-separable codes when $t\geq 3$}
\label{sec:general_lower_bound}

\begin{theorem}
\label{thm:general_lower_bound}
Let $n$ and $t$ be fixed integers such that $n\geq 2$ and $t\geq 3$. There exists a positive constant $\kappa$, depending only on $n$ and $t$, so that there is a $q$-ary $\overline{t}$-separable code of length $n$ with at least $\kappa q^{n/(t-1)}$ codewords for all sufficiently large integers $q$. 
\end{theorem}
The proof of Theorem~\ref{thm:general_lower_bound} will use random choice with expurgation, just as in the proof of Theorem~\ref{thm:2_lower_bound} above. However, we have to work harder to provide a good enough bound on the number of randomly chosen words we need to remove to form our code $\cC$: the straightforward approach of Theorem~\ref{thm:2_lower_bound} is no longer sufficient. Here is a brief summary of the method that is used.

As before, we choose words $\bc_1,\bc_2,\ldots,\bc_N\in F^n$ uniformly and independently at random. We want to identify (and ultimately remove elements from) subsets $X,X'\subseteq\{\bc_1,\bc_2,\ldots,\bc_N\}$ such that $|X'|\leq|X|\leq t$ and $\desc(X)=\desc(X')$. We fix positive integers $r$, $s$ and $s'$ and restrict ourselves to the cases when $|X\cap X'|=r$, $|X\setminus X'|=s$ and $|X'\setminus X|=s'$. For most choices of $r$, $s$ and $s'$ (choices such that $(r,s,s')$ lies in a certain set $J$ defined below) the approach of Theorem~\ref{thm:2_lower_bound} will work: we define a collection of appropriate events $F_{A,B,B'}$, where $A,B,B'\subseteq\{1,2,\ldots,N\}$ are disjoint subsets of indices of sizes $r$, $s$ and $s'$ respectively. However, we cannot do this when $r=t-1$ and $s=s'=1$: the probability of the events $F_{A,B,B'}$ is too large and will affect the leading term. So we need to work harder to deal with this case. Let $j$ and $j'$ be defined by $X\setminus X'=\{\bc_j\}$ and $X'\setminus X=\{\bc_{j'}\}$. It turns out that the probability that $\desc(X)=\desc(X')$ is dominated by situations when $\bc_j$ and $\bc_{j'}$ agree in many positions. So, rather than have events $E_{j,j'}$ corresponding to equality of $\bc_j$ and $\bc_{j'}$ as in Theorem~\ref{thm:2_lower_bound}, we define $E_{j,j'}$ to be the event that $\bc_j$ and $\bc_{j'}$ agree in many positions. We then define events $G_{A,j,j'}$ that require $\desc(X)=\desc(X')$ and also require that $E_{j,j'}$ does not occur. The remainder of the proof proceeds as in Theorem~\ref{thm:2_lower_bound}, although various calculations become a little more complicated.
\begin{proof}[Proof of Theorem~\ref{thm:general_lower_bound}]
Let $q$ be a fixed positive integer, and let $F$ be a set of cardinality $q$. Let $N$ be an integer to be chosen later. Choose words $\bc_1,\bc_2,\ldots,\bc_N\in F^n$ uniformly and independently at random.

For any $i,j\in\{1,2,\ldots,N\}$ with $i<j$, let $E_{i,j}$ be the event that $\bc_i$ and $\bc_j$ agree in $n/(t-1)$ or more components. There are less than $2^n$ ways of choosing a set $K$ of $n/(t-1)$ or more components; the probability that the codewords $\bc_i$ and $\bc_j$ agree on the components in $K$ is at most $q^{-n/(t-1)}$. So the event $E_{i,j}$ occurs with probability at most $2^nq^{-n/(t-1)}$. Let $R_{i,j}$ be the indicator random variable for $E_{i,j}$. The expected value of $R_{i,j}$ is at most $2^nq^{-n/(t-1)}$.

Let $A$, $B$ and $B'$ be disjoint subsets of $\{1,2,\ldots,N\}$ such that $|B'|\leq |B|$. Define the event $F_{A,B,B'}$ as follows. Let $X=\{\bc_i\mid i\in A\cup B\}$ and let $X'=\{\bc_i\mid i\in A\cup B'\}$. Then $F_{A,B,B'}$ is the event that $\desc(X)=\desc(X')$. Since we have
\[
\desc(\{\bc_i:i\in B\})\subseteq\desc(X)\subseteq\desc(X')
\]
when our event occurs, we see that for $i\in B$ the $k$th component of $\bc_i$ must agree with the $k$th component of one of the elements of $X'$. So the probability of the event $F_{A,B,B'}$ is bounded above by $((|A|+|B'|)/q)^{|B|n}$. Let $S_{A,B,B'}$ be the indicator random variable corresponding to $F_{A,B,B'}$. Then the expected value of $S_{A,B,B'}$ is at most $(|A|+|B|)^{|B|n}q^{-|B|n}$.

Define $J$ to be the set of triples $(r,s,s')$ of integers such that:
\begin{align*}
0&\leq r\leq t-1,\\
1&\leq s\leq t-r,\\
0&\leq s'\leq s,\\
2&\leq r+s,\text{ and }\\
(r,s,s')&\not=(t-1,1,1).
\end{align*}
We will only be interested in the events $F_{A,B,B'}$ where $(|A|,|B|,|B'|)\in J$. 

We now define events associated with the element $(t-1,1,1)$ we have removed from $J$. Let $A$, $\{j\}$ and $\{j'\}$ be disjoint subsets of $\{1,2,\ldots,N\}$ of cardinalities $t-1$, $1$ and $1$ respectively. Define the event $G_{A,j,j'}$ as follows. Let $X=\{\bc_i\mid i\in A\cup \{j\}\}$ and let $X'=\{\bc_i\mid i\in A\cup \{j'\}\}$. Then $G_{A,j,j'}$ is the event that $\bc_j$ and $\bc_{j'}$ agree in less than $n/(t-1)$ components and $\desc(X)=\desc(X')$. There are less than $2^n$ ways of choosing a set $K$ of components of size less than $n/(t-1)$, and $G_{A,j,j'}$ is the union of the sub-events where $\bc_j$ and $\bc_{j'}$ agree exactly in the components in $K$.  Define
\[
Y=\{\bc_i\mid i\in A\}.
\]
We have $X(k)$ is the union of $Y(k)$ and the $k$th component of $\bc_j$. Similarly, $X'(k)$ is the union of  $Y(k)$ and the $k$th component of $\bc_{j'}$. The condition $\desc(X)=\desc(X')$ is equivalent to the condition that $X(k)=X'(k)$ for all components $k$. So on the components $k\notin K$ where $\bc_j$ and $\bc_{j'}$ differ, the $k$th components of $\bc_j$ and $\bc_{j'}$ must both lie in $Y(k)$; this happens with probability at most $((t-1)/q)^2$. When $k\in K$, the codewords  $\bc_j$ and $\bc_{j'}$ agree in their $k$th components, and this happens with probability $1/q$. When $q$ is sufficiently large, $((t-1)/q)^2<1/q$, and so the probability of the event $G_{A,j,j'}$ is at most
\[
2^n(((t-1)/q)^2)^{n-n/(t-1)}q^{-n/(t-1)}< (2t)^{2n} q^{-(2-1/(t-1))n}.
\]
The above expression is an upper bound on the expectation of the indicator random variable $T_{A,j,j'}$ corresponding to $G_{A,j,j'}$.

Define $\mathcal{R}=\{(i,j):1\leq i<j\leq N\}$, so $|\mathcal{R}|\leq N^2$. For integers $(r,s,s')\in J$, define $\mathcal{S}_{r,s,s'}$ to be the set of triples $(A,B,B')$ of disjoint subsets of $\{1,2\ldots,N\}$ with $|A|=r$, $|B|=s$ and $|B'|=s'$, so $|\mathcal{S}_{r,s,s'}|\leq N^{r+s+s'}$. Finally define $\mathcal{T}$ to be the set of all triples $(A,j,j')$ where $A$, $\{j\}$ and $\{j'\}$ are disjoint subsets of $\{1,2,\ldots,N\}$ and $|A|=t-1$. We see that $|\mathcal{T}|\leq N^{t+1}$.

We are interested in excluding the events $E_{i,j}$ where $(i,j)\in \mathcal{R}$, the events $F_{A,B,B'}$ where $(A,B,B')\in\mathcal{S}_{r,s,s'}$ with $(r,s,s')\in J$, and the events $G_{A,j,j'}$ where $(A,j,j')\in \mathcal{T}$. So we define a random variable $Z$ by
\[
Z=\sum_{\{i,j\}\in \mathcal{R}}R_{i,j}+\sum_{(r,s,s')\in J}\sum_{(A,B,B')\in \mathcal{S}_{r,s,s'}}S_{A,B,B'}+\sum_{(A,j,j')\in\mathcal{T}}T_{A,j,j'}.
\]
We have computed upper bounds on the expectation of the indicator variables in this expression, and on the number of terms in each sum. From these bounds, and from linearity of expectation, we can deduce that $Z$ has expected value at most $b$, where we define
\begin{align}
\nonumber
b_1&=N^22^nq^{-(1/(t-1))n},\\
\nonumber
b_2&=\sum_{(r,s,s')\in J}N^{r+s+s'}(r+s)^{sn}q^{-sn},\\
\nonumber
b_3&=N^{t+1}(2t)^{2n} q^{-(2-1/(t-1))n},\text{ and}\\
\label{eqn:b_defn}
b&=b_1+b_2+b_3.
\end{align}

We now proceed as in Theorem~\ref{thm:2_lower_bound}. There is a choice of the values $\bc_i$ such that at most $b$ of the events occur. Fix this choice. Define a subset $I\subseteq\{1,2,\ldots,N\}$ by adding $i$ for each event $E_{i,j}$ that occurs, adding a single element $i\in A\cup B\cup B'$ for each event $F_{A,B,B'}$ that occurs and adding a single element $i\in A\cup\{j\}\cup\{j'\}$ for each event $G_{A,j,j'}$ that occurs. We then define $\cC=\{\bc_i\mid i\in \{1,2,\ldots,N\}\setminus I\}$.

If $\bc_j=\bc_{j'}$ for some $j\not=j'$ then the event $E_{j,j'}$ holds, so one of $j$ and $j'$ lies in $I$, thus $\bc_j$ or $\bc_j'$ was removed when $\cC$ was formed. So all the codewords in $\cC$ are distinct. In particular, $|\cC|\geq N-b$.

We now show that $\cC$ is a $\overline{t}$-separable code. Suppose for a contradiction that there are distinct non-empty subsets $X,X'\subseteq \cC$ with $|X|\leq t$, $|X'|\leq t$ and $\desc(X)=\desc(X')$. Without loss of generality, we may assume that $|X'|\leq |X|$. Define $Y=X\cap X'$. Let $r=|Y|$, $s=|X\setminus Y|$ and $s'=|X'\setminus Y|$. If $|X|=1$ then $|X'|=1$ and we have a contradiction since all codewords are distinct. So we may assume that $|X|=r+s\geq 2$. We now see that $(r,s,s')\in J\cup\{(t-1,1,1)\}$. If $(r,s,s')\in J$, then the event $F_{A,B,B'}$ occurs, where $A$, $B$ and $B'$ are the indices of the codewords in $Y$, $X\setminus Y$ and $X'\setminus Y$ respectively. But then at least one element of $A\cup B\cup B'$ lies in $I$, and so we have a contradiction in this case. Suppose now that $(r,s,s')=(t-1,1,1)$. Then $X=Y\cup\{\bc_j\}$ and $X'=Y\cup\{\bc_{j'}\}$ for distinct $j,j'\in\{1,2,\ldots,N\}$. Without loss of generality, we may assume that $j<j'$. If $\bc_j$ and $\bc_{j'}$ agree in $n/(t-1)$ or more components, then the event $E_{j,j'}$ occurs and so $j\in I$. Thus $\bc_j\not\in \cC$, and we have a contradiction. Now suppose that $\bc_j$ and $\bc_{j'}$ agree in less than $n/(t-1)$ components. Then the event $G_{A,j,j'}$ occurs, where $A$ is the set of indices corresponding to $Y$. But at least one of the indices $A\cup\{j\}\cup\{j'\}$ lies in $I$, which gives us a contradiction in this final case. Hence $\cC$ is a $\overline{t}$-separable code, as required.

Finally, we need to show that $N$ can always be chosen so that $|\cC|\geq \kappa q^{n/(t-1)}$ for all sufficiently large $q$. We set $N=\lfloor\epsilon q^{n/(t-1)}\rfloor$, where $\epsilon$ is a sufficiently small positive constant depending only on $n$ and $t$. We will always take $\epsilon<1$. Note that $|\cC|=N-b=\epsilon q^{n/(t-1)}-b-O(1)$. We will now show that when $N=\lfloor\epsilon q^{n/(t-1)}\rfloor$,
\begin{equation}
\label{eqn:b_bound}
b\leq \epsilon^2\kappa' q^{n/(t-1)}
\end{equation}
for some constant $\kappa'$ which depends only on $n$ and $t$. Proving this inequality establishes the theorem. For if we choose $\epsilon<1/\kappa'$ we see that
$\epsilon-\epsilon^2\kappa'$ is positive and so
\[
|\cC|\geq N-b\geq (\epsilon-\epsilon^2\kappa') q^{n/(t-1)}-O(1)\geq \kappa q^{n/(t-1)}
\]
for all sufficiently large $q$, where $\kappa$ is any positive constant less than $\epsilon-\epsilon^2\kappa'$.

Recall that $b=b_1+b_2+b_3$. We establish~\eqref{eqn:b_bound} by bounding each of $b_1$, $b_2$ and $b_3$ in turn. First,
\[
b_1=N^22^nq^{-n/(t-1)}\leq \epsilon^22^nq^{n/(t-1)}=\epsilon^2a_1q^{n/(t-1)}
\]
where $a_1$ depends only on $n$. 

We now turn to providing a bound for $b_2$. The constant $b_2$ is defined as a sum over $(r,s,s')\in J$. We begin by bounding a term in this sum.

Suppose that $(r,s,s')\in J$. We have that $r+s\leq t$ and $1\leq s$. Moreover, $r+s+s'\geq 2$ and so
\begin{align*}
N^{r+s+s'}(r+s)^{sn}q^{-sn}&\leq \epsilon^2(r+s)^{sn}q^{(r+s+s'-(t-1)s)n/(t-1)}\\
&\leq\epsilon^2t^{tn}q^{(r+s'-(t-2)s)n/(t-1)}.
\end{align*}

We claim that $r+s'-(t-2)s\leq 1$ for all $(r,s,s')\in J$. For if $s\geq 2$, then
\[
r+s'-(t-2)s\leq r+s-2(t-2)\leq t-2(t-2)=-t+4\leq 1,
\]
since $t\geq 3$. If $s=1$ and $r<t-1$, then
\[
r+s'-(t-2)s\leq (t-2)+s-(t-2)=1.
\]
Finally, if $s=1$ and $r=t-1$, then $s'=0$ and so
\[
r+s'-(t-2)s=t-1-(t-2)=1.
\]
So our claim follows, and
\[
N^{r+s+s'}(r+s)^{sn}q^{-sn}\leq \epsilon^2t^{tn}q^{n/(t-1)}
\]
whenever $(r,s,s')\in J$. Since $|J|\leq (t+1)^3$, we find that
\[
b_2=\sum_{(r,s,s')\in J}N^{r+s+s'}(r+s)^{sn}q^{-sn}\leq \epsilon^2(t+1)^3t^{tn}q^{n/(t-1)}=\epsilon^2a_2q^{n/(t-1)},
\]
where $a_2$ is a constant depending only on $n$ and $t$.

Finally, we provide a bound for $b_3$. We can see that
\[
b_3=N^{t+1}(2t)^{2n}q^{-(2-1/(t-1))n}\leq \epsilon^2(2t)^{2n}q^{((t+1)-2(t-1)+1)(n/(t-1))}.
\]
Now $(t+1)-2(t-1)+1=4-t\leq 1$ since $t\geq 3$. Hence
\[
b_3\leq \epsilon^2a_3q^{n/(t-1)}
\]
for some constant $a_3$ depending only on $n$ and $t$. 

The bounds above combine with~\eqref{eqn:b_defn} to show that~\eqref{eqn:b_bound} holds, where $\kappa'=a_1+a_2+a_3$. So the theorem follows.
\end{proof}

\paragraph{Acknowledgements} The author would like to thank Ying Miao for comments on an earlier version of this manuscript, and for an inspiring conference talk on the area at ALCOMA 15. The author would like to acknowledge EU COST Action IC1104 `Random Network Coding and Designs over $GF(q)$' for support that allowed him to attend this conference talk.

\end{document}